\documentclass[a4paper]{article}
\usepackage[comma,numbers,square,sort&compress]{natbib}
\usepackage{amsmath,amssymb,bm,amsthm}
\theoremstyle{definition} 
\usepackage{subfigure,pgf,tikz,pst-node}
\usetikzlibrary{arrows,shapes,automata,positioning,fit}
\usepackage{bm} 

\usepackage[ruled]{algorithm2e}
\usepackage[compatible]{algpseudocode}
\SetKwInOut{Parameter}{Parameters}

\newcommand{\R}{\mathbb{R}}
\newcommand{\N}{\mathcal{N}}
\newcommand{\G}{\mathcal{G}}
\newcommand{\I}{\mathbb{I}}
\newcommand{\Z}{\mathbb{Z}}

\newtheorem{lem}{Lemma}

\newtheorem{thm}{Theorem}
\newtheorem{ass}{Assumption}

\newtheorem{prob}{Problem}

\usepackage{pgf,tikz}
\usetikzlibrary{arrows,shapes,automata,positioning,fit}
\usepackage{pst-node}

\begin{document}
	
\title{Observer-based Leader-following Consensus for Positive Multi-agent Systems Over Time-varying Graphs}
	
\author{Ruonan Li, Yichen Zhang, Yutao Tang, and Shurong Li \footnote{This work was supported by National Natural Science Foundation of China under Grants 61973043. R. Li, Y. Zhang, Y. Tang, and S. Li are all with the School of Artificial Intelligence, Beijing University of Posts and Telecommunications, Beijing 100876. P.\,R. China. Emails: nanruoliy@163.com, zhangyc930@163.com, yttang@bupt.edu.cn, lishurong@bupt.edu.cn}}
	
\date{}
	
\maketitle
	
{\noindent\bf Abstract}:  This paper addresses the leader-following consensus problem for discrete-time positive multi-agent systems over time-varying graphs. We assume that the followers may have mutually different positive dynamics which can also be different from the leader. Compared with most existing positive consensus works for homogeneous multi-agent systems,  the formulated problem is more general and challenging due to the interplay between the positivity requirement and high-order heterogeneous dynamics. To solve the problem, we present an extended version of existing observer-based design for positive multi-agent systems. By virtue of the common quadratic Lyapunov function technique, we show the followers will maintain their state variables in the positive orthant and finally achieve an output consensus specified by the leader. A numerical example is used to verify the efficacy of our algorithms.


\section{Introduction}

Leader-following consensus is a hot topic in multi-agent coordination and has been studied for many years due to its wide applications in multi-robot systems and sensor networks. In such problems, a special agent called leader is designed whose motion is independent of other agents and the rest of agents (i.e., followers) are expected to track the output trajectory of the leader. Numerous interesting results have been derived for different agent dynamics in both continuous time and discrete time scenarios \cite{hong2006tracking, ren2008distributed, ni2010leader, su2011cooperative, zhu2015event, huang2017consensus, liu2021discrete}. 

Among various classes of dynamic systems, positive systems, in which the state and output trajectories are restricted to the positive orthant, have received much attention during the past few years \cite{farina2000positive, blanchini2015switched}.  Since many descriptor variables are usually constrained to be nonnegative due to their intrinsic characteristics or physical laws,  positive systems have been found in different applications in systems biology, chemical process, and communication positive systems. Recently, many efforts have been made to reach a consensus for positive multi-agent systems. This is partially motivated by the fact that many simple dynamic models such as integrators and first-order lags with positive gains, as well as their series/parallel connections, are all positive. For instance, the authors studied on the positive consensus problem for single-integrator multi-agent systems in \cite{valcher2013stabilizability}. It is shown that this problem can  be equivalently converted to some positive stabilization problem for certain multi-input multi-output systems but subject to information constraints. Similar idea has also been exploited in  \cite{valcher2017consensus, liu2019positivity, bhattacharyya2022positive, sun2017stabilization, liu2022positivity} including a recent attempt \cite{cao2021positive}, where multiple Lyapunov functions have been used to deal with average dwell-time switching.  Nevertheless, most existing positive consensus papers are focused on homogeneous positive multi-agent system. That is, all agents share the same system matrices,  which definitely limits the potential applications of proposed positive consensus algorithms. Thus, it is natural for us to ask that whether the above arguments remain effective when the agents have positive but nonidentical dynamics.  Although similar issues have been discussed for standard multi-agent systems in the literature, whether and how to solve the problem for heterogeneous positive multi-agent systems is still unclear.

Based on aforementioned observations, we focus on the leader-following positive consensus problem for heterogeneous  multi-agent systems under time-varying communication graphs. In other words, all agents are allowed to possess mutually different dynamics from each other in both state-space dimensions and system matrices. As a result, the conventional treatment by Kronecker product can not be implemented. We have to seek new design methodologies to solve the positive consensus problem.  

Recently, observer-based designs are proven to be very effective in handling the heterogeneity of multi-agent systems  \cite{su2011cooperative,  huang2017consensus, tang2016coordination, liu2018distributed, liu2021redueced, cao2021positive, tang2018distributed, tang2020optimal}. Motivated by these results, we aim at distributed  observer-based design to achieve the expected leader-following positive consensus goal.  However, observer-based control even for a single positive plant is nontrivial in contrast to standard linear systems as shown in \cite{rami2011positive}. When facing multiple positive agents, we have to deal with extra technical difficulties from heterogeneous agent dynamics and time-varying interactions simultaneously. Thus, how to ensure the positivity of agent trajectories as well as the asymptotic tracking performance by distributed observer-based  controllers is definitely much more challenging than that for standard multi-agent systems or the homogeneous positive counterpart. 

To solve the problem,  we first constructively present a distributed positive observer as a combination of existing distributed results for standard linear multi-agent systems and centralized positive observer design for a single positive plant.    With a specialized initialization, we show its positivity and convergence  by rigorous analysis. After that, we utilize such positive observers to resolve the formulated leader-following positive consensus problem by both state feedback and output feedback controllers.  To our knowledge, this might be the first attempt to solve the leader-following consensus problems for general heterogeneous positive multi-agent systems in this setting. 

The main contributions of this paper can be summarized as follows: 
\begin{itemize}
	\item This paper extends existing positive consensus results for homogeneous multi-agent systems to the leader-following case for a group of heterogeneous multi-agent systems. The followers are allowed to have mutually different agent dynamics while the leader can generate more general positive collective motions other than a static consensus.  
	\item A novel distributed observer-based design is developed for  positive multi-agent systems. The present design can be taken as a distributed counterpart of the centralized rules in \cite{roszak2009necessary}  and extends existing observer-based protocols for standard multi-agent systems by further imposing a positivity requirement. 
	\item  By virtue of common Lyapunov function techniques, we show that the expected positive consensus can be reached under the proposed controllers over time-varying communication graphs. Compared with the recent result in \cite{cao2021positive}, we do not require any average dwell-time condition and allow arbitrary switching among connected graphs. 
\end{itemize}
 
The rest of this paper is organized as follows. In Section \ref{sec:pre},  we provide some preliminaries and the formulation of our positive consensus problem. Then, the main results are presented in Section \ref{sec:main} along with further discussions in Section \ref{sec:discuss}. After that, a simulation example is  given to illustrate the effectiveness of our algorithms in Section \ref{sec:simulation}. Finally,  conclusions are given in Section \ref{sec:conclusion}.

\section{Preliminary}\label{sec:pre}

In this section, we first introduce some preliminaries on graph theory and positive systems for the following analysis, and then present the formulation of our problem.

\subsection{Graph theory}
A directed graph (digraph) is employed  to represent the information flow among the agents. A directed graph $\mathcal {G}=(\mathcal {N},\, \mathcal{E}\,)$ is defined with node set $\mathcal{N}=\{1,\, {\dots},\, n\}$ and edge set $\mathcal{E}\subset \mathcal{N}\times \mathcal{N}$.  $(i,\,j)\in \mathcal{E}$ denotes an edge leaving from node $i$ and entering node $j$.  A directed path in digraph $\mathcal {G}$ is an alternating sequence $i_{1}e_{1}i_{2}e_{2}{\cdots}e_{k-1}i_{k}$ of nodes $i_{l}$ and edges $e_{m}=(i_{m},\,i_{m+1})\in\mathcal{E}$ for $l=1,\,2,\,{\dots},\,k$. If there exists a path from node $i$ to node $j$,  then node $j$ is said to be reachable from node $i$. If any two nodes are reachable from each other, the digraph is said to be strongly connected.  The neighbor set of agent $i$ is defined as $\mathcal{N}_i=\{j\colon (j,\,i)\in \mathcal{E} \}$ for $i=1,\,\dots,\,n$.  Specially, a directed graph is said to be bidirected if $(i,\,j) \in \mathcal{E}$ implies $(j,\,i)\in \mathcal{E}$.  The adjacency matrix $\mathcal{A}\in \mathbb{R} ^{n\times n}$ corresponding to digraph $\mathcal{G}$ is defined as follows: $a_{ij}=1$ if $(j,\,i)\in \mathcal{E}$; otherwise, $a_{ij}=0$.  The Laplacian $L=[l_{ij}]\in \mathbb{R}^{n\times n}$ for this digraph $\mathcal{G}$ is defined as $l_{ii}=\sum_{j\neq i} a_{ij}$ and $l_{ij}=-a_{ij} (j\neq i)$. 
It can be found that the Laplacian has a trivial eigenvalue with an eigenvector ${\bf 1}_n$. 

\subsection{Positive system} 
 Let $\mathbb{R}_+$ be the set of nonnegative real numbers and $\mathbb{Z}_+$ the set of all nonnegative integers.  Denote the set of all $m\times n$ matrices with each entry in $\mathbb{R}_+$ by  $\mathbb{R}_+^{m\times n}$  and let ${\bf 1}$ (or ${\bf 0}$) represent an all-one (or all-zero) matrix or vector with proper dimensions.    We say such matrices are nonnegative whenever $a_{ij}  \geq0$ and adopt the notation $A\ge {\bf 0}$.  If, in addition, $A $ has at least one positive entry, we say $A$ is positive with $A>{\bf 0}$, i.e., $A\geq {\bf 0}$ and $A \neq {\bf 0}$.  If all the entries are positive,   $A$ is strictly positive with $A\gg {\bf 0}$. The nonnegativity, positivity, and strictly positivity of vectors can be defined likewise.  A matrix $A\in \mathbb{R}^{n\times n}$ is said to be Schur  if all its  eigenvalues have norm less than one.  We use the symbols $A \succ {\bf 0}$ and $ A \succeq \, {\bf 0}$ to describe the fact that $A$ is positive definite and positive semidefinite matrix, respectively. We will reverse the direction of inequality symbols if $-A$ satisfies one of the aforementioned properties.  
 
 Consider the following linear discrete-time system:
 \begin{align} \label{sys:def}
 	\begin{split}
 		x(k+1) &=A(k)x(k)+B(k) u(k)\\
 		y(k) &=C(k) x(k),\quad k\in \Z_+
 	\end{split}
 \end{align}
 where $x$ is the  $n$-dimensional state vector, $u$ is the  $p$-dimensional control input, and $y$ is the $l$-dimensional output vector. Here, $A(k),\,B(k),\,C(k)$ are time-varying system matrices with compatible dimensions. System \eqref{sys:def} is said to be (internally) positive if for any nonnegative initial condition $x(0)$ and any nonnegative input $u(0),\,u(1),\,\dots $,  it holds that  $x(k)\geq {\bf 0}$ and $y(k)\geq {\bf 0}$ for $k\in \Z_+$.  For constant matrices $A,\,B,\,C$,  we say the pair $(A,\,B)$ is positively stabilizable if there exists a gain matrix $K\in {\mathbb R}^{n\times p}$ such that $A+BK$ is nonnegative and Schur stable. This system is said to be positively detectable if the pair $(A^\top,\, C^\top)$ is positively stabilizable.

 The following lemma is a discrete-time counterpart of Lemma {\rm VIII.1} in  \cite{angeli2003monotone}, which can be practically established. We omit its proof to save space. 
 \begin{lem}\label{lem:positive}
 	System \eqref{sys:def} is positive if and only if the matrices $A(k),\,B(k),\,C(k)$ are nonnegative for any $k\in \Z_+$.
 \end{lem}

 
\subsection{Problem formulation} 

Consider the leader-following coordination problem for a group of $N+1$ positive agents with $N$ followers of the following form:
\begin{align} \label{sys:follower}
	\begin{split}
		x_i(k+1) &= A_i x_i(k) +B_i u_i(k)  \\
		y_i(k) &= C_i x_i (k),\quad k \in \Z_+
	\end{split}
\end{align}
where $x_i \in \mathbb{R}^{n_i}$, $u_i \in \mathbb{R}$, and $y_i \in \mathbb{R}$ are the state, input, and output of agent $i$ with $i\in \bar\N\triangleq \{1,\,\dots,\, {N}\}$. Moreover, we assume the system \eqref{sys:follower} is both positively stabilizable and positively detectable for each $i\in \bar\N$. 

Suppose the leader (i.e., agent $0$) is of the following form:
\begin{equation} \label{sys:leader}
	\begin{split}
		x_0(k+1) &= A_0 x_0(k)\\
		y_0(k) &=C_0 x_0(k),\quad k \in \Z_+
	\end{split}
\end{equation}
with internal state $x_0 \in \R^{n_0}$ and output $y_0\in \R$.  We assume the leader is positive and marginally stable as follows.   
\begin{ass} \label{ass:exo}
	Matrices $A_0$ and $C_0$ are both nonnegative. Moreover, there exists a diagonal matrix $P_0\succ {0}\in \R^{n_0\times n_0}$ such that $A_0^\top P_0 A_0-P_0 \preceq{\bf 0}$. 
\end{ass}

Assume only a few followers know the information of the leader while the followers are allowed to communicate with others. Let $e_i(k)=y_i(k)- y_0(k)$ for $i\in \bar{\N}$ and $k\in \Z_+$. We aim at distributed rules for these followers to track the leader in the sense that $\lim_{k\to \infty} e_i(k)=0$ for each $i\in \bar{\N}$. 

For this purpose, we use digraph $\G(k)=\{\N,\,\mathcal{E}(k)\}$ with node set $\N=\{0\}\cup \bar \N$ to describe the allowed information flow among this multi-agent system at time $k\in \Z_+$. When agent $i$ can receive information from agent $j$ at time $k$, there is a directed edge from node $j$ to node $i$ in digraph ${\G}(k)$. Then, the induced subgraph $\bar {\G}(k)=\{ \bar{\N},\, \bar {\mathcal{E}}(k) \}$ represents the information flow among the followers.  Moreover, we consider the case when the communication topology among the agents are chosen from a finite set $\{ \G _1,\,\dots,\, \G_P\}$ with ${\cal P}=\{1,\,2,\,\dots,\, P\}$. To be clear, we introduce a function $\sigma(k) \colon \Z_+ \to {\cal P}$ to describe the communication graphs with ${\cal G}(k)=\mathcal{G}_{\sigma(k)}$.   

The following assumption has been widely used in the multi-agent coordination literature \cite{hong2006tracking,tang2016coordination,cao2021positive}.
\begin{ass}\label{ass:graph}
	For each $p\in \cal{P}$, digraph $\mathcal{G}_p$ contains a spanning tree with node $0$ as the root and subgraph $\bar{\cal G}_p$ is bidirected and strongly connected.  
\end{ass}

Under this assumption, the leader has an independent motion from the followers, while each follower is supposed to have bidirectional communications with the neighboring followers for its controller. A follower can also receive information from the leader if there exists a directed edge from node $0$ to it.  For digraph $\mathcal{G}_p$, we delete the first row and first column of its Laplacian and denote the rest submatrix by $H_p \in \mathbb{R}^{N\times N}$. Let $B_p\triangleq \mbox{diag}(b_{1p},\cdots, b_{Np})$ with $b_{ip}=1$ if node $i$ can get information from node $0$ and $b_{ip} =0$ otherwise. It can be verified that $H_p=L_p+B_p$ is symmetric. Moreover, it is also positive definite under Assumption \ref{ass:graph} according to Lemma 3 in \cite{hong2006tracking}. Since only finite number of graphs fulfill this assumption, we let  $\bar\lambda=\max_{p \in \mathcal{P}} \{ \lambda_{\max}(H_p)\}$ and $\underline\lambda=\min_{p \in \mathcal{P}} \{ \lambda_{\min}(H_p)\}$ with $\lambda_{\max}(\cdot)$ (or $\lambda_{\min}(\cdot)$) the maximal (or minimal) eigenvalue for some given positive definite matrices. These constants are verified to be well-defined and strictly greater than zero.

We are ready to formulate the leader-following positive consensus problem for heterogeneous multi-agent systems as follows: 
\begin{prob}\label{prob}
	Given a multi-agent system composed of $N$ positive followers \eqref{sys:follower} and a positive  leader \eqref{sys:leader} with graph $\mathcal{G}(k)$, find distributed controllers of the following form: 
	\begin{align}\label{ctrl:nominal}
		\begin{split}
			u_i (k) &=f_i(x_j(k),\,\eta_j(k))\\
			{\eta}_i(k+1)&=g_i(x_j(k),\,\eta_j(k)), \quad j \in \N_i\cup i	
		\end{split}
	\end{align} 
	with proper smooth functions $f_i$, $g_i$ and a compensator $\eta_i\in \R^{n_{{\eta}_i }}$ such that, for any initial point $x_i(0)\geq {\bf 0}$ and $x_0(0)\geq 
	{\bf 0}$, the composite system \eqref{sys:follower}--\eqref{ctrl:nominal} satisfies 
	\begin{itemize}
		\item[1)] The trajectory of $x_i(k)$ is  nonnegative, i.e., $x_i(k)\geq {\bf 0} $ for any $k\in \Z_+$ and $i\in \bar{\N}$.
		\item[2)] The tracking error $e_i(k)$ converges to $0$ as $ k \to \infty$.
	\end{itemize}	
\end{prob}

The formulated leader-following  consensus problem has been intensively studied  in the past few years for standard multi-agent systems. In the formulated problem, we further impose an extra  positivity requirement on the agent trajectories. Since we consider both heterogeneous agent dynamics and time-varying communication graphs, this intrinsic nonlinearity makes our problem more technically challenging than existing results for homogeneous multi-agent system over fixed graphs \cite{valcher2017consensus, liu2019positivity, bhattacharyya2022positive}. 



\section{Main Result} \label{sec:main}
 
To solve problem, we will constructively present observer-based controllers in this section to solve the formulated problem. The main design idea is illustrated in  Fig.~\ref{fig:framework}.  We will first construct effective distributed positive observers of the leader for each follower. Then, we embed this distributed observer into some designed decentralized tracking controllers (possibly with local state observers) to solve our leader-following positive consensus problem. 
\begin{figure}
	\centering
	\includegraphics[width=0.80\textwidth]{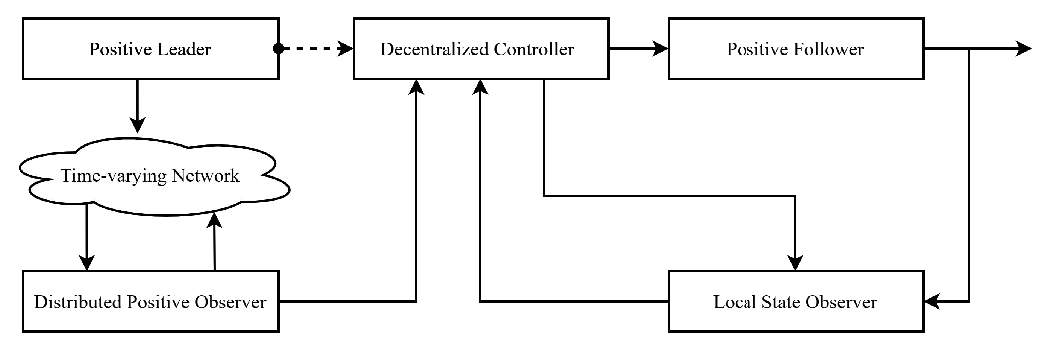} 
	\caption{Illustration of our control scheme.}
	\label{fig:framework}
\end{figure} 

Before the main result, we introduce an extra assumption on agent models.
\begin{ass} \label{ass:regeq}
	For each $i\in \bar \N$, there exist constant matrices $X_{i}\in \R_+^{n_i\times n_0}$ and $U_i\in \R_+^{1 \times n_0}$ satisfying the following matrix equations:
	\begin{align}\label{eq:regulator}
	\begin{split}
	A_i X_{i} +B_i U_i - X_{i} A_0&= {\bf 0}\\ 
	C_i X_{i} -C_0&={\bf 0}
	\end{split}		
	\end{align}
\end{ass}

This assumption can be understood as the solvability of the regulator equations for plant \eqref{sys:follower} with exosystem \eqref{sys:leader} in the terminology of output regulation, which plays a crucial role in our following designs. For homogeneous multi-agent systems, we can trivially choose $X_i=\I_{n_i}$ and $U_i={\bf 0}$ to confirm this assumption.  Similar conditions have been widely used in the multi-agent coordination literature \cite{su2011cooperative,  tang2018distributed, wang2010distributed}.

Let us start from the distributed observer part. Since we consider positive consensus problem, it is natural to develop distributed positive observers for each agent. For this purpose, we consider the following candidate:
	\begin{equation} \label{sys:observer-1}
		w_i (k+1) = A_0 w_i(k) + \mu A_0 w_{vi}(k), \, i\in\bar \N,\, k\in \Z_+
	\end{equation}
where $w_i \in \R^{n_0}$ is the estimate on the leader's state $x_0$ for agent $i$. Here, $w_{vi}(k)\triangleq\sum_{j=0}^N a_{ij}(k) (w_j(k)-w_i(k))$ for $i\in \bar \N$ with $w_{0}(k)=x_0(k)$ represents the neighboring information at agent $i$ and $\mu>0$ is a constant to be determined later. This observer is motivated by existing works for standard linear multi-agent systems \cite{su2011cooperative, tang2016coordination}. However, the positivity requirement might be violated if we only consider the choice of $\mu$ as in distributed observer-based designs for standard multi-agent systems.  To tackle this issue, we let $\tilde{w}_i(k)=x_0(k)-w_i(k)$ for $i\in \bar\N$ and extend the centralized initialization idea in \cite{roszak2009necessary} to the distributed case.

Here is a lemma to summarize the effectiveness of \eqref{sys:observer-1} to asymptotically reproduce the trajectory of the leader that fulfills the positive constraint.  

\begin{lem} \label{thm:1observer}  
Suppose Assumptions \ref{ass:exo} and \ref{ass:graph} hold. Let $ 0 <\mu < \frac{1}{N+1}$ and set $w_i(0)={\bf 0}$ for $i\in \bar \N$.  Then, for each $i\in \bar \N$ and $k\in \Z_+$, along the trajectory of \eqref{sys:observer-1},  $w_i(k) \geq   {\bf 0}$ for any initial condition 
	${x}_0(0)\geq  {\bf 0}$ and $\tilde{w}_i (k)$ exponentially converges to ${\bf 0}$ as $k \to \infty$.
\end{lem}
\begin{proof}
	To prove this lemma, we firstly put the observer \eqref{sys:observer-1} into a compact form: 
	\begin{align} \label{compactob1}
	\begin{split}
	{w}(k+1) &=[(\I_N- \mu H_{\sigma(k)} ) \otimes  A_0] w(k) +
	\mu (\Lambda_{\sigma(k)} \otimes A_0)({\bf 1}_{N}\otimes w_0(k) ) 
	\end{split}  
	\end{align}
	where $\Lambda_{\sigma(k)} \triangleq \mbox{diag}(a_{10}(k),\,\dots,\,a_{N0}(k))\geq   {\bf 0}$ and $ w(k)=\mbox{col}(  w_1(k),\,\dots,\,  w_N(k))$. Under Assumption \ref{ass:exo}, $\mu (\Lambda_{\sigma(k)} \otimes A_0)({\bf 1}_{N}\otimes w_0(k))$ is positive. To prove the positivity of the corresponding trajectory of \eqref{sys:observer-1},  we only need to prove the nonnegativity of its system matrix $(\I_N- \mu H_{\sigma(k)} ) \otimes  A_0$. For this purpose, we just have to show $\I_N-\mu H_{\sigma(k)} \geq {\bf 0}$ since  $A_0\geq 0$ under Assumption \ref{ass:exo}. Expanding this matrix as follows:
	\begin{align} \label{principal}
	\begin{bmatrix}
	1-\mu h_{11} & \mu h_{12} & \dots & \mu h_{1N}  \\
	\mu h_{21} & 1-\mu h_{22}  & \dots& \mu h_{2N} \\
	\dots & \dots& \dots& \dots \\
	\mu h_{N1} & \mu h_{N2}  & \dots& 1-\mu h_{NN} 	\end{bmatrix}
	\end{align}
	where $h_{ij}$ is the element of time-varying matrix $-H_{p}$ at its $i$-th row and $j$-th column.   Then, it is sufficient for us to ensure that the diagonal elements are nonnegative for any time, which naturally holds under the lemma condition. Then, we recall that $w(0)={\bf 0}$ and obtain the positivity of $w(k)$ at any time according to Lemma \ref{lem:positive}.  
	
	Next, we consider the following error system to prove the convergence of $\tilde w_i(k)$:   
	\begin{equation} \label{sys:error1w}
	\tilde{w}(k+1) = [(\I_N- \mu H_{\sigma(k)}) \otimes A_0] \tilde{w}(k)
	\end{equation}
	where $\tilde w(k)=\mbox{col}(\tilde w_1(k),\,\dots,\,\tilde w_N(k))$.  Since $(\I_N -\mu H_{\sigma(k)}) \otimes A_0$ is nonnegative as shown above, it follows that $\tilde w_i(k)\geq {\bf 0}$ due to the fact that $\tilde w_{i}(0)\geq {\bf 0}$. Thus, we only have to focus on the stability of \eqref{sys:error1w}. 
	
	To this end, we choose $V_{\tilde w}(k)=  \tilde {w}^\top(k) ({\I}_N\otimes P_0)\tilde {w}(k)$ as a candidate of common Lyapunov function. It is quadratic and positive definite. We check its evolution with respect to $k$.  Suppose $\sigma(k)=p$.  It follows then
		\begin{align*}
		\Delta_V(k)&\triangleq V_{{\tilde {w}}}(k+1)-V_{{\tilde {w}}}(k)\\
		&=\tilde{w}(k)^\top  [(\I_N- \mu H_p)^2 \otimes (A_0^\top P_0 A_0)] \tilde{w}(k)-\tilde{w}(k)^\top ({\I}_N\otimes P_0) \tilde{w}(k)
		\end{align*}
	
	Since $H_p$ is positive definite under Assumption \ref{ass:graph}, there must be a unitary matrix $U_p$ such that $U_p^{\top} H_p U_p =D_p$ where $D_p$ is a diagonal matrix with the eigenvalues of $H_p$ on the diagonal. Let $\check w(k)=(U_p^\top \otimes {\mathbb I}_{n_0})\tilde w(k)$. It follows then
	\begin{align*}
	\Delta_V(k) =\check{w}(k)^\top  [(\I_N- \mu D_{p})^2 \otimes (A_0^\top P_0 A_0) ] \check{w}(k)-\check{w}(k)^\top ({\I}_N\otimes P_0) \check{w}(k)
	\end{align*}
	Since $\I_N- \mu D_{p}$ is a diagonal matrix with diagonal elements $1-\mu \lambda_i(H_p)$ and $\lambda_i(H_p)$ the $i$-th ordered eigenvalue of $H_p$. Under the theorem conditions, we can verify that 
	\begin{align*}
	-1<1-\mu \bar \lambda \leq 1-\mu \lambda_i(H_p)<1-\mu \underline\lambda<1
	\end{align*}  
	Thus, there must be a constant $\bar \rho>0$ such that  $1-\mu \lambda_i(H_p) \leq \bar \rho<1$   for any $p$. This implies that 
	\begin{align*}
	\Delta_V(k)&\leq \check{w}(k)^\top  \{[(\I_N- \mu D_{p})^2-\I_N] \otimes  P_0 \} \check{w}(k)\\
	&\leq -(1-\bar \rho^2) V_{\tilde w}(k)
	\end{align*}
	Since the above arguments hold for any $k\in \Z_+$ and $p\in \mathcal{P}$, $V_{\tilde w}(k)$ is indeed a common Lyapunov function for the switched error system \eqref{sys:error1w}.  From this, we  conclude the  exponential convergence of  ${\tilde w}(k)$ towards ${\bf 0}$ as $k \to  \infty$ and complete the proof. 
\end{proof}

Compared with plenty of distributed observers for standard multi-agent systems, we utilize a special initialization to ensure the positivity of this observer \eqref{sys:observer-1}. This enables us a positive and convergent estimate of the expected trajectory determined by the leader to facilitate our following design. 

\begin{algorithm}
	\caption{Distributed State Feedback Controller for Positive Consensus}  \label{alg:state-feedback}
	\SetKwInput{Parameter}{Parameters}
	\SetKwInput{Initial}{Initialization}
	\SetKwInput{Update}{Update}
	\KwData{Follower \eqref{sys:follower}, leader \eqref{sys:leader}, graph  $\mathcal{G}_p$,  switching signal $\sigma(t)$\;}
	
	\Parameter{Determine $\mu$, $K_{1i}$, and $K_{2i}$ for $i\in \bar \N$: 
		
		a. Find a feasible   diagonal matrix $Q_i> {\bf 0}$   such that  
		\begin{align*}
		&(B_i^\top Q_i^{-1}B_i)\mathbb{I}_{n_i}-B_i B_i^\top Q_i^{-1}  \geq {\bf 0}\\ 
		&A_i Q_i A_i^\top-B_iB_i^\top- Q_i \prec {\bf 0} 
		\end{align*}
		
		b. Solve the equations \eqref{eq:regulator} to obtain $X_i\geq {\bf 0}$ and $U_i\geq {\bf 0}$\;
		
		c. Set $0 <\mu < \frac{1}{N+1}$, $K_{1i}=-(B_i^\top Q_i^{-1}B_i)^{-1}B_i^\top Q_i^{-1} A_i$, $K_{2i}=U_i-K_{1i}X_i$ \;
	}
	\Initial{Set $w_i(0)={\bf 0}$ for $i\in \bar \N$\;} 
	\Update{Communicate with the neighbors and update local data according to 
		\begin{align*}
		w_{vi}(k)&=\sum\nolimits_{j=0}^N a_{ij}(k) (w_j(k)-w_i(k))\\
		u_i(k)&=K_{1i} x_i(k) + K_{2i} w_i(k)\\
		{x}_i(k+1) &= A_ix_i(k)+ B_iu_i(k) \\
		{w}_i(k+1)&= A_0 w_i(k) + \mu  A_0 w_{vi}(k),\quad i\in \bar \N, \quad k\in \Z_+
		\end{align*}
	}	
	\KwResult{Output $y_i(k)$ and tracking error $y_i(k)-y_0(k)$.}
\end{algorithm}

Next, we complete the whole design to solve our formulated  positive consensus by  distributed  controllers.  We first choose a matrix $K_{1i} \in \R^{1\times n_i}$ such that $A_i+B_i K_{1i}$ is nonnegative and Schur. Such a gain matrix indeed exists under our assumptions and can be taken as a nonpositive one according to Lemma 6 in \cite{roszak2009necessary}. Let $K_{2i}=U_i-K_{1i}X_i$. We incorporate the  preceding  distributed positive observer \eqref{sys:observer-1} into a full-information rule and present the following state feedback controller for agent $i$:
\begin{align}\label{con-obs-1}
	\begin{split}
		u_i(k) &=K_{1i} x_i(k) + K_{2i} w_i(k)\\
		{w}_i(k+1) &= A_0 w_i(k) +\mu A_0 w_{vi}(k),\quad i\in \bar \N
	\end{split}
\end{align}
where $w_i(0)={\bf 0}$ and $\mu$ is specified as in Lemma \ref{thm:1observer}.  This controller is of the form \eqref{ctrl:nominal} and indeed distributed relying  on local communication and computation at agent $i$.  
For a better illustration, we summarize the design procedure as Algorithm \ref{alg:state-feedback} including how to construct corresponding gain matrices. 

The main result of this paper is summarized  as follows.
\begin{thm}\label{thm:main}
	Suppose Assumptions \ref{ass:exo}--\ref{ass:regeq} hold. The formulated leader-following positive consensus problem  for multi-agent systems \eqref{sys:follower} and \eqref{sys:leader} can be solved by distributed controllers of the form \eqref{con-obs-1}. 
\end{thm} 
\begin{proof}
	Putting the agent dynamics \eqref{sys:follower} and controller \eqref{con-obs-1}  together, we first obtain the following composite system: 
	\begin{align} \label{composite}
		\begin{split}
			{x}_i(k+1) &= (A_i+B_iK_{1i})x_i(k)+ B_iK_{2i} w_i(k)\\  
			{w}_i(k+1) &= A_0 w_i(k) + \mu  A_0 w_{vi}(k)\\
			y_i(k)&=C_i x_i(k), \quad i \in \bar\N
		\end{split}
	\end{align}	

Under theorem conditions, we can verify that $K_{2i}=U_i-K_{1i}X_i\geq {\bf 0}$ and $w_i(k)\geq {\bf 0}$ by Lemma \ref{thm:1observer}. It follows that $K_{2i} w_i(k)$  is nonnegative for any $k\in \Z_+$.  Since the above system matrix  $A_i+B_iK_{1i}$ at agent $i$ is positive by the choice of $K_{1i}$, we view   the term $B_iK_{2i} w_i$ as a control input of the $x_i$-subsystem and have $x_i(k)\geq{\bf 0}$ for any $x_i(0)\geq {\bf 0}$ and $k\in \Z_+$. Then, we are left to show that the convergence of the tracking error $e_i(k)$ towards $0$ as $k \to \infty$. 
	
For this purpose, we let $\tilde x_i (k)=x_i(k)-X_ix_0(k)$ and obtain the following tracking error system for agent $i$: 
	\begin{align}\label{sys:error-final-1}
		\begin{split}
			{\tilde x}_i(k+1) &= (A_i+B_iK_{1i})\tilde x_i(k)+ B_iK_{2i} \tilde w_i(k)\\
			{\tilde w}(k+1) &=[(\I_N-\mu H_{\sigma(k)}) \otimes A_0]\tilde{w}(k)\\
			e_i(k)&=C_i\tilde x_i(k), \quad i\in \bar \N
		\end{split}
	\end{align}
We consider the $\tilde x_i$-subsystem and regard  $B_iK_{2i} \tilde w_i(k)$ as an input. Since $A_i+B_iK_{1i}$ is Schur by theorem conditions, the $\tilde x_i$-subsystem must be input-to-state stable with state $\tilde x_i$ and input. Recalling Lemma \ref{thm:1observer}, we know that $\tilde w(k)$ is  convergent. Thus, we can conclude the convergence of $\tilde x_i$ towards ${\bf 0}$ as $k \to \infty$, which implies $\lim_{k\to \infty} e_i(k)= { 0}$. The proof is thus complete.  
\end{proof}

In many circumstances, the state $x_i$ may be unavailable for our controllers due to the high cost of sensors or physical constraints. With the proposed distributed positive observer \eqref{sys:observer-1}, we can further present an output feedback extension for  \eqref{con-obs-1}  as follows:
\begin{align} \label{con-obs-2}
	u_i(k)&= K_{1i} \eta_i(k) + K_{2i} w_i(k) \nonumber\\
	\eta_i(k+1)&=(A_i- K_{3i}C_i) \eta_i(k)+ B_i u_i(k)+K_{3i} y_i(k) \nonumber\\
	{w}_i(k+1) &= A_0  w_i(k) + \mu A_0 w_{vi}(k), \quad i\in \bar \N
\end{align}
where $K_{3i}\in \R_+^{n_i\times 1}$ is a matrix such that $(A_i-K_{3i}C_i)$ is nonnegative and Schur. Similar to \eqref{con-obs-1}, we set $\eta_i(0) = {\bf 0}$ to ensure the positivity of corresponding estimate $\eta_i$ for $x_i$. 

The effectiveness of this controller is given as follows.
\begin{thm}\label{thm:main-2}
	Suppose Assumptions \ref{ass:exo}--\ref{ass:regeq} hold.  The formulated leader-following positive consensus problem for positive multi-agent system \eqref{sys:follower} and \eqref{sys:leader} can be solved by distributed output feedback controllers of the form \eqref{con-obs-2}.
\end{thm}
\begin{proof}
	The existence of $K_{3i}$ is guaranteed by the positive detectability of $(A_i,\,C_i)$. Since the convergence part is similar to that of Theorem \ref{thm:main}, we only have to show $x_i(k) \geq {\bf 0}$ under the controller \eqref{con-obs-2}. 
	
	Let $\bar x_i(k)= x_i(k)-\eta_i(k)$ and put the full composite system  as follows:
	\begin{align*}
		{x}_i(k+1) &= (A_i+B_iK_{1i})x_i(k)+\Xi_i(k)\\
		\bar x_i(k+1) &=(A_i-K_{3i}C_i) \bar x_i(k)\\
		{w}_i(k+1) &= A_0 w_i(k) + \mu A_0 w_{vi}(k)\\
		y_i(k)&=C_i x_i(k),\quad i\in \bar \N
	\end{align*}
	where $\Xi_i(k)\triangleq -B_iK_{1i}\bar x_i(k) + B_iK_{2i} w_i(k)$. By the choice of $K_{3i}$, the matrix $A_i-K_{3i}C_i$ is nonnegative and Schur. Note that $\bar x_i(0) = x_i(0)-\eta_i(0) = x_i(0) \geq {\bf 0}$,  we have $\bar x_i(k)\geq {\bf 0}$ from Lemma \ref{lem:positive}. Again, we consider the $x_i$-subsystem and view $\Xi_i(k)$ as its input. Since $-B_iK_{1i}\bar x_i(k) \geq {\bf 0}$ by the  nonpositivity of $K_{1i}$, the term $\Xi_i(k)\geq {\bf 0}$ for any time $k$. As a consequence, $x_i(k) \geq {\bf 0}$ for any $k$ by Lemma \ref{lem:positive}.  The proof is thus complete. 
\end{proof}
	
\section{Discussion} \label{sec:discuss}

In this section, we provide some remarks to show how our observer-based design extends previously published results. 

\subsection{The nonexistence of common copositive Lyapunov functions} 

Copositive linear Lyapunov functions have been widely used in positive consensus analysis for fixed graphs. When the communication graph is time-varying, we may expect some common copositive Lyapunov function. Unfortunately, such common copositive Lyapunov functions may not exist.  This might be one of the technical obstacles hindering us from general positive consensus results over time-varying graphs.

Here is  a counterpart example.  Suppose we have two followers and the matrix $A_0=1$.  The communication graphs are chosen from digraphs $\G_1$ and $\G_2$ with
	\begin{align*}
		H_1=\begin{bmatrix}
			2&-1\\-1&1
		\end{bmatrix},\quad H_2=\begin{bmatrix}
		1&-1\\-1&2
		\end{bmatrix}
	\end{align*}

Consider a distributed observer of the form \eqref{sys:observer-1}. We put down the following matrix
	\begin{align*}
		\begin{bmatrix}
			H_1-\I_2 & H_2-\I_2
		\end{bmatrix}=
		\begin{bmatrix}
			1&-1&0&-1\\
			-1&0&-1&1
		\end{bmatrix}
	\end{align*}
and find that the sum of its first and fourth columns are $\mbox{col}(0,\,0)$. According to Theorem 1 in \cite{fornasini2012stability},  no common copositive Lyapunov function can be found to confirm the conclusion.  By contrast, we can verify all assumptions in Lemma \ref{thm:1observer} and obtain the convergence and positivity of \eqref{sys:observer-1} using quadratic-type common  Lyapunov function as in its proof.  

\subsection{Relationships with existing results}

While most observer-based results are derived for standard linear multi-agent systems \cite{hong2006tracking, su2011cooperative,  huang2017consensus, tang2016coordination, liu2018distributed, liu2021redueced, cao2021positive,tang2018distributed, tang2020optimal}, whether and how this approach is still effective for positive multi-agent systems heavily depends upon the positivity of the observers.  To tackle this issue, we utilize a special initialization  first proposed in \cite{roszak2009necessary}  and extend it to the distributed case. This enables us a positive and convergent estimate of the expected trajectory determined by the leader. To our knowledge, Lemma \ref{thm:1observer} gives us the first affirmative answer to the existence of distributed positive observers under mild assumptions. 

It is remarkable to interpret Lemma \ref{thm:1observer} as a leaderless positive consensus problem for homogeneous multi-agent systems with identical system matrix $A_0$ and input matrix $\I_{n_0}$. As such, the derived results are consistent with existing positive consensus results in the literature \cite{valcher2017consensus, liu2019positivity,Jason2022positive, bhattacharyya2022positive, sun2017stabilization}.  Moreover, we consider time-varying communication graphs, which definitely extends the aforementioned results for only fixed graphs. Compared with the recent work \cite{cao2021positive} under a similar setting, by virtue of common Lyapunov functions, we do not need any average dwell-time condition assumption and allow arbitrary switching among connected graphs. 

With the aid of this distributed positive observer \eqref{sys:observer-1}, we provide two different kinds of distributed  controllers  to solve the  formulated leader-following positive consensus problem.  When $N=1$, the considered leader-following problem reduces to the conventional positive tracking problem for single positive plant. Thus, it can be taken as a distributed extension of these centralized results \cite{farina2000positive, roszak2010multivariable, wang2019output, yang2021proportional}. Moreover, different from most existing positive consensus results \cite{valcher2013stabilizability, valcher2017consensus, liu2019positivity, bhattacharyya2022positive,Jason2022positive,Jason2020robust, wu2018observer}, our distributed observer-based designs do not require the agents to be homogeneous. In fact, under some standard matrix equations assumption, the agent dynamics (including the followers and the leader) are allowed to be mutually different from each other. We will also illustrate this point by an example in Section \ref{sec:simulation}.  
	
\section{Simulation} \label{sec:simulation}

In this section, we consider a six-agent system with positive dynamics to illustrate the effectiveness of our preceding positive consensus algorithms. 

Suppose the system matrices of this multi-agent system are given as follows: 
	 \begin{align*}
	 	&A_0=\begin{bmatrix} 1&0&0\\ 0&0&1 \\ 0&1&0 \end{bmatrix}, \quad C_0=\begin{bmatrix}1&1&0.5\end{bmatrix}\\
		&A_1=A_2=\begin{bmatrix} 0.1&0.2&0.3 \\ 0.1& 0.3& 0.2 \\  0.1& 0.3&0.1 \end{bmatrix}, \quad B_1=B_2=\begin{bmatrix} 1\\1\\2 \end{bmatrix}, \quad C_1=C_2  =\begin{bmatrix} 1 & 0 &2\end{bmatrix} \\
		&A_3=A_4=\begin{bmatrix} 0&1\\1 & 0 \end{bmatrix}, \quad B_3=B_4=\begin{bmatrix} 1\\2 \end{bmatrix},\quad C_3=C_4  =\begin{bmatrix}1 &0 \end{bmatrix}  \\
		&A_5= 0.5, \quad  B_5=1,\quad  C_5  =1
	\end{align*}

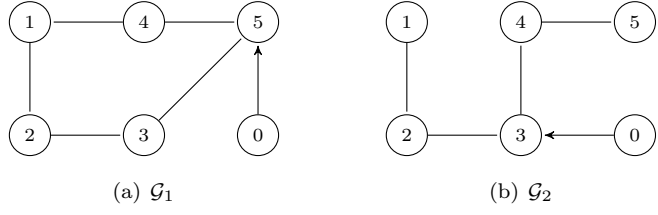
\begin{figure}
	\centering
	\subfigure[$\mathcal{G}_1$]
	{\centering
		\begin{tikzpicture}[shorten >=1pt, node distance=1.5 cm, >=stealth',
			every state/.style ={circle, minimum width=0.4cm, minimum height=0.3cm}]
			\node[align=center,state](node1) {\scriptsize 1};
			\node[align=center,state](node2)[below of=node1]{\scriptsize 2};
			\node[align=center,state](node3)[right of=node2]{\scriptsize 3};
			\node[align=center,state](node4)[above of=node3]{\scriptsize 4};
			\node[align=center,state](node5)[right of=node4]{\scriptsize 5};
			\node[align=center,state](node0)[right of=node3]{\scriptsize 0};
			\path[-] (node1) edge (node2)
			(node2) edge (node3)
			(node4) edge (node1)
			(node4) edge (node5)
			(node3) edge (node5);
			\path[->] (node0) edge (node5);
	\end{tikzpicture}}\qquad \qquad 
	\subfigure[$\mathcal{G}_2$]
	{\centering
		\begin{tikzpicture}[shorten >=1pt, node distance=1.5 cm, >=stealth',
			every state/.style ={circle, minimum width=0.4cm, minimum height=0.3cm}]
			\node[align=center,state](node1) {\scriptsize 1};
			\node[align=center,state](node2)[below of=node1]{\scriptsize 2};
			\node[align=center,state](node3)[right of=node2]{\scriptsize 3};
			\node[align=center,state](node4)[above of=node3]{\scriptsize 4};
			\node[align=center,state](node5)[right of=node4]{\scriptsize 5};
			\node[align=center,state](node0)[right of=node3]{\scriptsize 0};
			\path[-] (node1) edge (node2)
			(node2) edge (node3)
			(node3) edge(node4)
			(node4) edge (node5);
			\path[->] (node0) edge (node3);
	\end{tikzpicture}}
	\caption{The communication graphs in our example.} \label{fig:graph1}
\end{figure}

\begin{figure}
	\centering
	\subfigure[$||\tilde w_i(k)||$]{
		\includegraphics[width=0.44\textwidth]{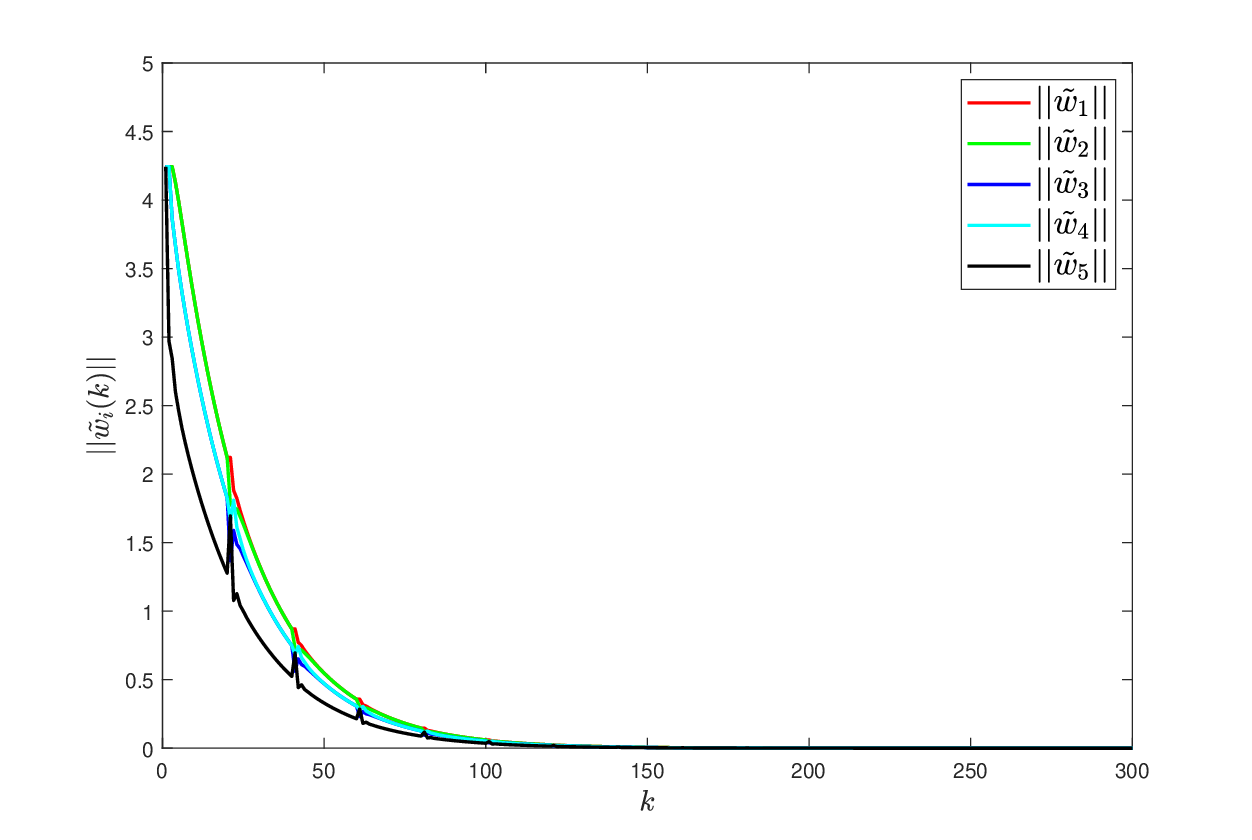} }
	\quad 
	\subfigure[${w}_2$]{
		\includegraphics[width=0.44\textwidth]{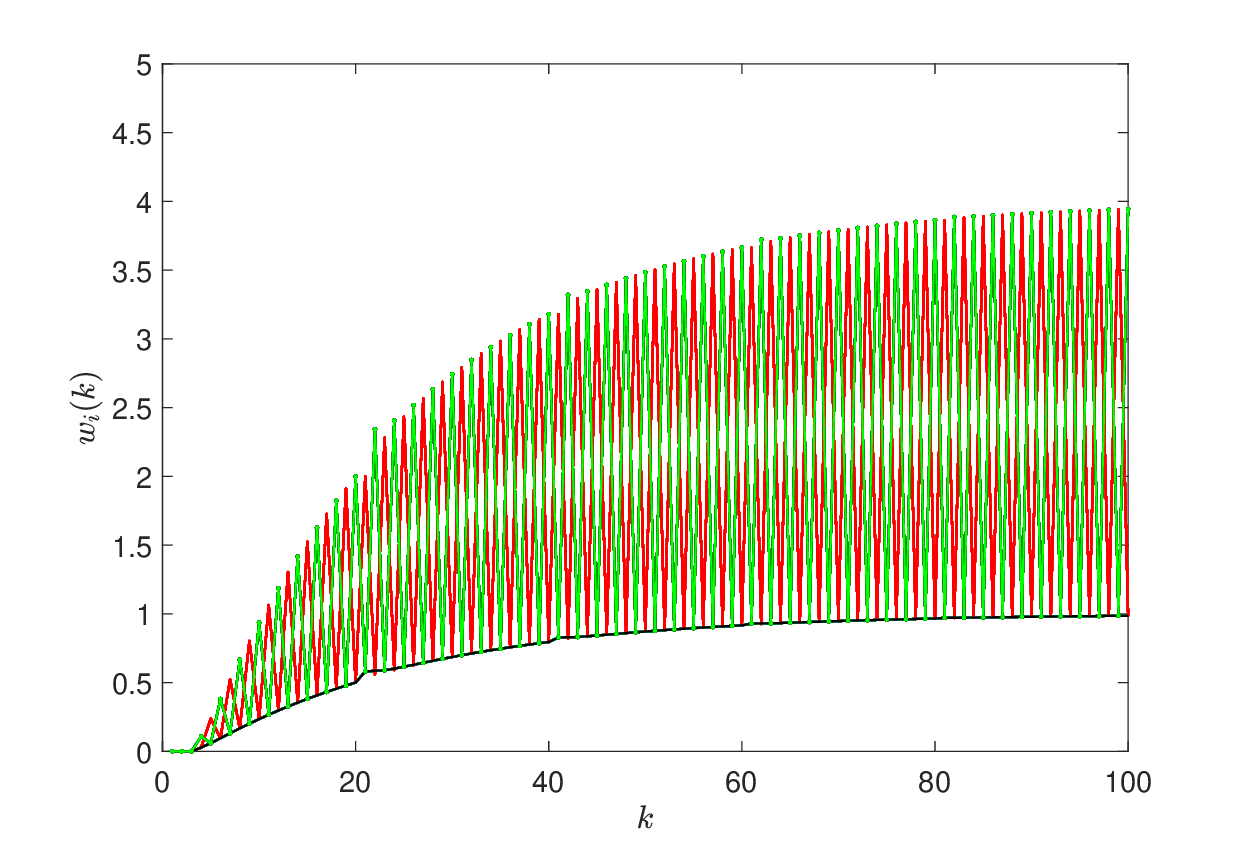}
		}
	\caption{Performance of the distributed positive observer \eqref{sys:observer-1}.}
	\label{fig:simu:observer}
\end{figure}

\begin{figure}
	\centering
	\includegraphics[width=0.80\textwidth]{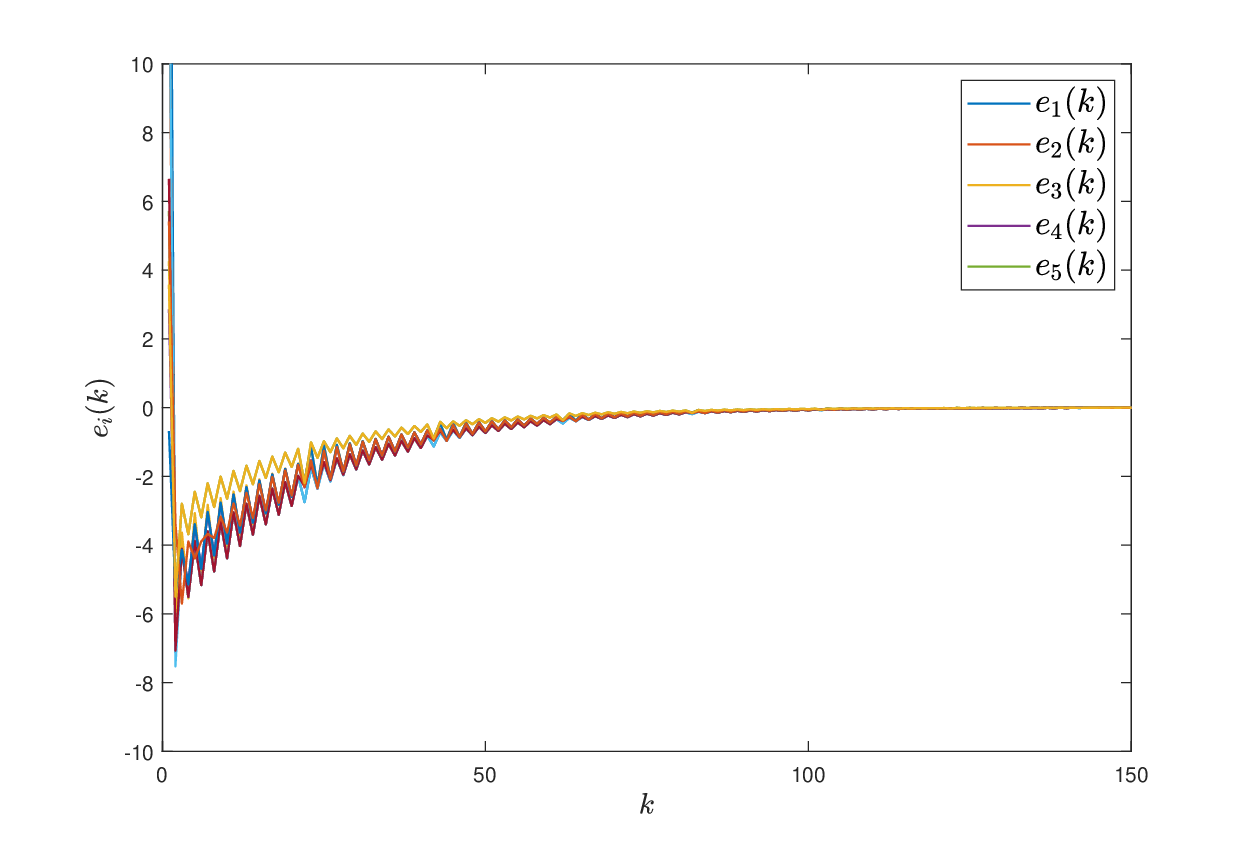} 
	\caption{Tracking error under output feedback controller \eqref{con-obs-2}.}
	\label{fig:simu:output-error}
\end{figure} 

By Lemma \ref{lem:positive}, these agents are all positive. Moreover, the followers are positively stabilizable and detectable according to Theorem 8  in   \cite{roszak2009necessary}. We can also confirm Assumption \ref{ass:exo} by letting $P_0=\I_3$.  Meanwhile, we can numerically solve the regulator equations and obtain 
  \begin{align*}
 	&	X_1=X_2=\begin{bmatrix}
 		0.2939  & 0.2456   & 0.1952\\
 		0.2873  &  0.2486  &  0.1823\\
 		0.3531  &  0.3772  &  0.1524
 	\end{bmatrix},\quad U_1=U_2=\begin{bmatrix}   0.1011 &  0.0078  & 0.1439 \end{bmatrix}\\
 	&	X_3=X_4=\begin{bmatrix}
 	   1.0000 &1.0000 & 0.5000\\
 	1.0000 & 0.5000 & 1.0000
 	\end{bmatrix},\quad U_3=U4=\begin{bmatrix}0& 0 &  0 \end{bmatrix} \\
 	& X_5=\begin{bmatrix}
 	 1.0000 & 1.0000 & 0.5000
 	\end{bmatrix},\quad U_5=\begin{bmatrix} 0.5000  &   0&  0.7500\end{bmatrix}
 \end{align*}
Hence, Assumption \ref{ass:regeq} also holds.  Suppose the communication topology is alternatively switching  between two digraphs $\G_1$ and $\G_2$ given in Fig.~\ref{fig:graph1} every $20$ steps.  Then, we can resort to Theorems \ref{thm:main} and \ref{thm:main-2}  to solve the corresponding leader-following positive consensus  by a distributed controller of the form \eqref{con-obs-1} or \eqref{con-obs-2}. 

In the simulation, we solve the leader-following positive consensus problem by output feedback controller \eqref{con-obs-2}.  We follow the procedure in Algorithm \ref{alg:state-feedback} to choose the parameters. Solving the corresponding linear matrix inequalities, we choose $K_{11}=K_{12}=-\begin{bmatrix}  0.0667  & 0.1833   &0.1167 \end{bmatrix}$, $K_{13}=K_{14}=-\begin{bmatrix} 0.4000&0.2000\end{bmatrix}$, and $K_{15}= -0.5000$ for positive stabilization.  Similarly, we let $K_{31}=K_{32}=\begin{bmatrix} 0.0600 & 0.1600 & 0.1000 \end{bmatrix}^\top $,  $K_{33}=K_{34}=\begin{bmatrix} 0&1 \end{bmatrix}^\top $, and $K_{35}=0.5000$ for the local positive observer part. Finally, we let $\mu=0.3$ and $w_i(0)={\bf 0}$ for the distributed positive observer \eqref{sys:observer-1}. The rest of initial conditions are randomly generated between $0$ and $12$.

We first show the convergence performance of our observer in Fig.~\ref{fig:simu:observer} where the estimation error $||\tilde w_i||$ converges to ${  0}$ quickly. We also list the full trajectory of $w_2(k)$ as an example to confirm the positivity of this observer. Under the distributed output feedback controller \eqref{con-obs-2}, the tracking error of each agent is given in Fig.~\ref{fig:simu:output-error}, where the leader-following consensus is reached as we expected. We then list the profile of  $x_2(k)$ in Fig.~\ref{fig:simu:output-posivity}. It can be found that all components of $x_2(k)$ are kept to be nonnegative.  These observations confirm the effectiveness of our preceding designs to solve the leader-following positive consensus problem for heterogeneous multi-agent systems under time-varying communication graphs.

\begin{figure}
	\centering
	\includegraphics[width=0.84\textwidth]{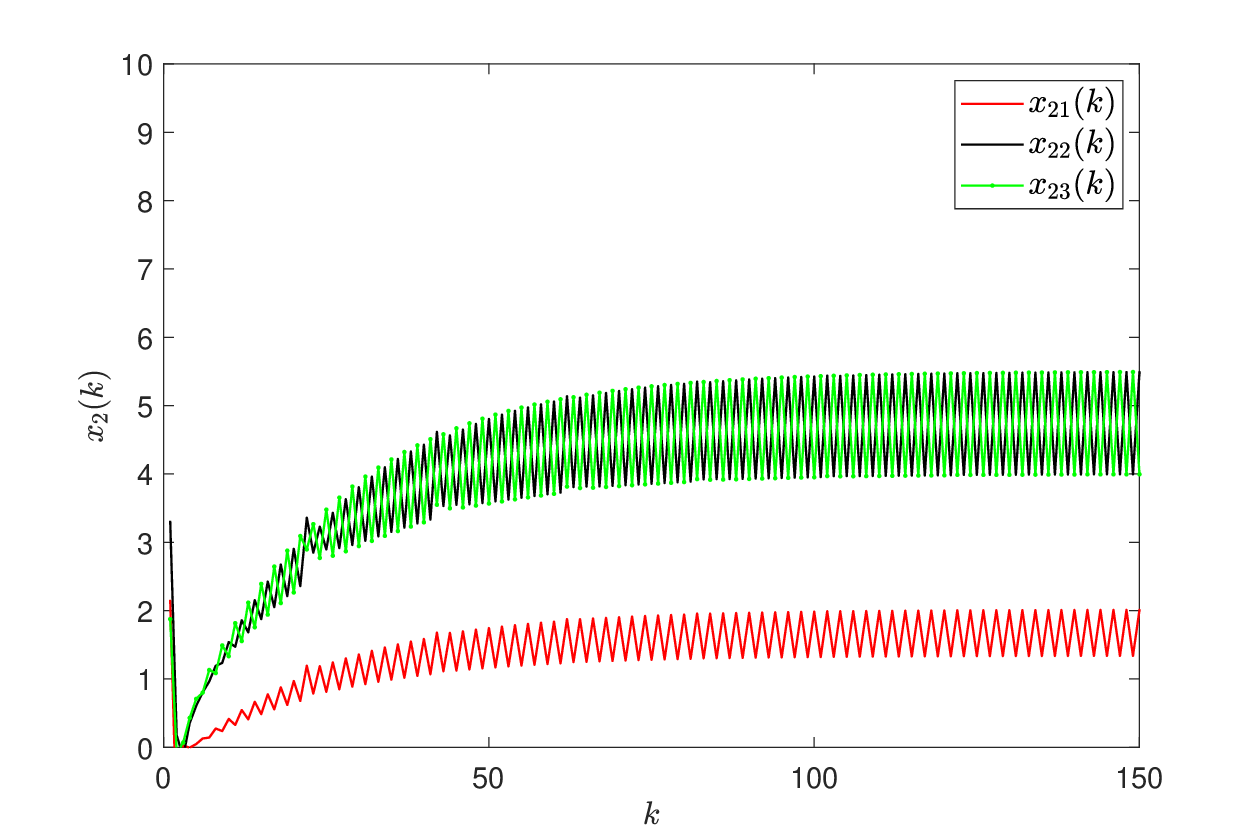} 
	\caption{Positivity of $x_2(k)$ under output feedback controller \eqref{con-obs-2}.}
	\label{fig:simu:output-posivity}
\end{figure}

\section{Conclusion} \label{sec:conclusion}
 
This paper has studied the leader-following positive consensus problem for a group of discrete-time linear multi-agent systems. To handle the heterogeneous agent dynamics and time-varying graphs, we have developed applicable distributed positive observers for each follower and then proposed two different kinds of distributed controllers via incorporating the estimate of positive observers. The followers are shown to be able to track the reference leader as well as fulfilling the positive state constraints under time-varying communication graphs. In the future, we may consider the same problem but for uncertain positive multi-agent systems.

\section{Acknowledgment}

This work was supported by National Natural Science Foundation of China under Grant 61973043 and BUPT innovation and entrepreneurship support program.
	
\bibliographystyle{IEEEtran}
\bibliography{canre}

\end{document}